\documentclass{article}
\include{matthew.sty}

\usepackage{tikz}
\usetikzlibrary{positioning,calc}
\usetikzlibrary{decorations.pathmorphing, arrows.meta}
\usetikzlibrary{shapes.geometric}
\tikzset{harp/.style={arrows = {[harpoon]<->[harpoon]}}}
\definecolor{p0Blue}{RGB}{0, 90, 169} 
\definecolor{p1Red}{RGB}{230, 0, 26}
\usepackage{pgfplots}
\pgfplotsset{compat=1.18}
\usepgfplotslibrary{fillbetween}
\newcommand{\minwalks}{\operatorname{minwalks}}
\newcommand{\walks}{\operatorname{walks}}
\newcommand{\kw}[1]{#1\operatorname{-walks}}
\newcommand{\vm}{V_{\operatorname{Max}}}
\newcommand{\vn}{V_{\operatorname{Min}}}
\newcommand{\emx}{E_{\operatorname{Max}}}
\newcommand{\emn}{E_{\operatorname{Min}}}
\newcommand{\SetOf}[2]{\left\{#1 \colon #2\right\}}
\newcommand{\val}{\operatorname{val}}
\newcommand{\B}{\mathcal{B}}

\begin{document}

\title{Connecting Strategy Improvement, the Simplex Algorithm and Lopsidedness
}


\author{Matthew Maat\\
\textit{University of Twente}}
\date{}

\maketitle

\begin{abstract}
The strategy improvement algorithm for mean payoff games and parity games is a local improvement algorithm, just like the simplex algorithm for linear programs. Their similarity has turned out very useful: many lower bounds on running time for the simplex algorithm have been created, by first constructing lower bounds for strategy improvement in a game, and then transforming these into a linear program. However, earlier connections between these algorithms required constructing an intermediate Markov decision process, which is not always possible. We prove a formal, direct connection between the two algorithms, showing that many variants of strategy improvement for parity and mean payoff games are truly an instance of the simplex algorithm, under mild nondegeneracy assumptions. The linear programming formulation of these games leads to what we call hypercubelike polyhedra. These polyhedra allow us to show combinatorial properties of various games on graphs. In particular, we show a connection between strategy sets of games and lopsided sets. 
\end{abstract}

\section{Introduction}\label{intro}
The simplex algorithm and strategy improvement are two algorithms which are similar in many aspects. Not only are they similar in their core ideas as local improvement algorithms, but they also both allow some flexibility in choosing which improvements to make first. As a result of that, both algorithms carry a long-standing open question, which asks whether there exists a variant that can solve their respective problem in (strongly) polynomial time. 

The simplex method or simplex algorithm is one of the most well-known algorithms for solving linear programs (LPs). In essence, it starts at a vertex of the feasible polyhedron, and in each iteration moves to an adjacent vertex with better objective value. Continuing this process, it is guaranteed to find the optimum of the linear program. Longest shortest path problems (LSPs) are two player zero-sum games played on the vertices of a graph. They can be seen as special cases of mean payoff games (MPGs), although they are equally hard to solve as MPGs. As with many other games of this type, they lend themselves to be solved with a variant of strategy improvement: this algorithm starts with some initial strategy, and applies improving changes (w.r.t. some valuation) to the strategy, which amount to swapping one edge in the strategy for another edge. It does so until it arrives at the optimal strategy. 

For both the simplex method and strategy improvement there is some freedom of choice within the algorithm: one can still choose which adjacent vertex to move to in the LP, or which edge to change in the strategy. The choice is made by a so-called pivot rule or improvement rule, respectively. Since the invention of the simplex method, the question has been open whether there is a pivot rule for which the algorithm runs in (strongly) polynomial time. This question is of particular interest, since a positive answer would solve two major open questions, namely Smale's 9\textsuperscript{th} problem \cite{smale1998mathematical} and the polynomial Hirsch conjecture \cite{Santos2012Counterexample}. 

Similarly, the question whether a polynomial-time algorithm exists for solving parity games and mean payoff games has been open for decades. The complexity of these problems has attracted great interest, as solving these games is one of the few problems in NP$\cap$coNP \cite{Zwick1996TheGraphs} for which no polynomial-time algorithm is known. In recent years, it has been found that there are quasi-polynomial algorithms for solving parity games \cite{calude_deciding_2017}, and that mean payoff games can be solved in smoothed polynomial time \cite{loff_smoothed_2024} and in polynomial time on average for certain probability distributions \cite{allamigeon_tropical_2014}.

For many natural choices of pivot rules for the simplex algorithm, it was shown that they take exponential time on certain deformed hypercubes, see for example \cite{amenta_deformed_1996} for an overview. The main idea of the constructions is to perturb the hypercube such that the simplex method visits all of its vertices. However, some randomized and memory-based pivot rules are designed to counteract these deformed hypercube constructions, making it very difficult to show that they take exponential time. To analyze worst-case performance of these more complicated rules, another technique is needed.  

In these cases, a more versatile technique has been used for some pivot rules. The main idea is to start with an instance that is difficult for strategy improvement in an LSP (usually from a parity game), and then to turn the LSP into a Markov decision process (MDP). This turns a run of strategy improvement into a run of the similar policy iteration algorithm. Under some conditions, policy iteration on MDPs is equivalent to the simplex algorithm running on a linear program formulation of the MDP. This gives a method of constructing hard instances for the simplex algorithm: first, one constructs a lower bound for strategy improvement in an LSP, and then if possible, one turns it into a linear program via an MDP. In this manner, one gains a lower bound for the simplex method. This turns out to be a powerful method of constructing these lower bounds. Most notably, the random-edge and random-facet rules have been shown to not be polynomial with this technique \cite{friedmann2011randomfacet,friedmann_errata_2014,friedmann_subexponential_2011}, just as the memory-based Zadeh's rule \cite{disser_exponential_2023} and Cunningham's rule \cite{avis_exponential_2017}. 

However, the main challenge with the above approach lies in the step where one turns an LSP into an MDP: this can only be done for a few hand-crafted examples, as there is no general reduction. To the best knowledge of the author, a formal connection between strategy improvement for LSPs/MPGs and the simplex method has not been made before. While there is a known reduction from mean payoff games to LPs, along with connections to tropical LPs (see related work below), in these reductions there is no good interpretation for strategies, and thus no clear connection to strategy improvement.

\paragraph{Contributions.} In this paper we formally establish the connection between strategy improvement in LSPs and the simplex method. We present a new formulation of LSPs as linear programs, that can be constructed for any LSP under mild nondegeneracy conditions. Because of existing reductions between games, this is also an LP formulation of mean payoff and parity games. The main new feature of this LP is that performing the simplex algorithm on it is equivalent to using strategy improvement in the LSP (and therefore, also for certain mean payoff and parity games). It also turns out that other values from the LP, like reduced costs, have interpretations in the related games. 

The main idea of this LP formulation is closely related conceptually to tropical linear programs; however, we do not use any tropical algebra. To be more precise, the linear program we use is related to the dual of the linear program in \cite{schewe_parity_2009} (which, in turn, is closely related to tropical LPs).

Formalizing this reduction has three useful implications. First of all, it allows us to use a much broader class of games to construct polyhedra, and thus have a more powerful tool for worst-case analysis of the simplex algorithm. As a concrete example, the reduction from this paper has recently been used by the author and others \cite{DBLP:conf/stacs/DisserLMM26} to show that the class of index-based pivot rules equipped with some memory has superpolynomial worst-case runtime. 
Secondly, the results from this paper simplify the proofs for many pivot rules that they take a superpolynomial number of steps: this is because since we can skip the reduction to MDPs for the existing lower bound constructions.

Lastly, we gain some additional insights into the structure of these games from the LP formulation. To be precise, the LP polyhedron has a structure\footnote{possibly after slight reformulation} that we call hypercubelike. We use this structure to prove that the sets of winning strategies of many games on graphs are lopsided sets. In turn, the rich structure of lopsided sets yields us some properties of winning sets in games on graphs.

\paragraph{Related work.} There are some connections known between LSPs and linear programming. It has been shown that solving parity games and mean payoff games can be reduced to solving a linear program \cite{schewe_parity_2009}. The idea behind this reduction is closely related to tropical linear programs. Many problems in tropical linear programs are in fact equivalent to solving mean payoff games \cite{akian_tropical_2012}. It was even shown that mean payoff games can be solved in polynomial time on average for certain probability distributions, using a tropical equivalent of the simplex algorithm \cite{allamigeon_tropical_2014}.

The hypercubelike polyhedron we define looks structurally like a part of a combinatorial cube. A characterization of combinatorial cubes can be found in \cite{adiprasito_realization_2023}. Some other related polyhedra are prismatoids and 2-level polytopes. Prismatoids have been studied in relation to polytope diameters in \cite{santos2012embedding}, and the relation of 2-level polytopes with combinatorial optimization was studied in \cite{aprile20182}.

The concept of lopsided sets was introduced to understand the intersection of convex sets with different orthants \cite{lawrence_lopsided_1983}. These sets appear in many different contexts, have many equivalent definitions and are also known as extremal set systems or ample classes, among others. A more detailed overview of lopsided sets is given in \cite{bandelt_combinatorics_2006,CHALOPIN20221}.

\section{Preliminaries}\label{sec:preliminaries}
\paragraph{The simplex method} is one of the most well-known algorithms for solving linear programs (LPs). We consider LPs of the following form:
\begin{eqnarray}
    \min & & c^Tx\nonumber\\
    s.t. & & Ax=b\label{eq:LP}\\
    & & x\geq 0\nonumber
\end{eqnarray}
Where $A\in \mathbb{Z}^{m\times n}$, $b\in \mathbb{Z}^m$ and $c\in\mathbb{Z}^n$. We have $m\leq n$, and $A$ is assumed to be of full rank. The simplex algorithm chooses a so-called basis $\B\subseteq [n]$, which gives the indices of $m$ linearly independent columns of $A$. We will denote the submatrix of $A$ that consists of these columns by $B$ and also refer to this submatrix as the basis. Let $x_B$ denote the elements of the vector $x$ indexed by $\B$, and let $x_N$ denote the remaining elements. Every basis has an associated basic solution, which is the unique vector $x$ that satisfies the equalities of (\ref{eq:LP}) and has $x_N=0$. If the basic solution $x$ is also feasible for the LP (\ref{eq:LP}) - which happens if it is nonnegative - we call it a \emph{basic feasible solution} or BFS. If the feasible region of the LP is a simple polyhedron, there is a one-to-one correspondence between BFSs and vertices of the polyhedron.

The simplex algorithm has two phases: the first phase consists of finding a feasible basis. In this paper, we ignore the first phase and assume we are given a feasible basis. The second phase starts with some BFS, and then in each iteration moves to a BFS whose vertex is connected to the first BFS by an edge, and such that the objective value decreases in every iteration. This amounts to replacing one element of the basis by another element. If there is no improvement possible, we have found the optimal solution\footnote{We may also detect unboundedness. However, the LPs we consider are guaranteed to have an optimal solution.}.

Another way to view the simplex method is by the so-called simplex tableau, which is the following matrix:
\[
H=
\begin{bmatrix}
c^T-c_B^TB^{-1}A & c_B^TB^{-1}b\\
-B^{-1}A & B^{-1}b
\end{bmatrix}
\]
In this matrix, we can read the values of the basic variables and the objective value in the rightmost column. The topmost row displays the reduced costs: these are negative if and only if the related variable is an improving variable. Exchanging a variable in the basis can be done simply by row operations in the matrix.

\paragraph{Directed graphs} are denoted by $G=(V,E)$. A graph is called simple if it has no loops and no two edges with the same tail and head. The graphs we consider are almost simple; the difference is that we allow one node to have a loop. A $v_1,v_{k+1}$-walk is a sequence of edges and vertices $v_1,e_1,v_2,e_2,\ldots,e_k,v_{k+1}$ such that $e_i=(v_i,v_{i+1})$ for all $i$. For convenience of notation, we often consider a walk just as a sequence of edges. We call a walk closed if its starting vertex is the same as its end vertex. A $v_1,v_{k+1}$-path is a walk in which all vertices are distinct. A cycle is a closed walk in which only the starting and end vertex are the same, and all other vertices are distinct. We consider edge-weighted graphs, and denote the edge weights by $w:E\to \mathbb{Z}$. We also extend the weight function to walks: if the walk $W$ is denoted by $e_1,e_2,\ldots,e_k$, then $w(W):=\sum_{i=1}^k w(e_i)$.

\paragraph{Mean payoff games and parity games} are two types of two-player games played on the vertices of a simple directed graph. The players are called Player 0 and Player 1, and the game graph is of the form $(V_0\cup V_1,E)$, where the vertices of $V_i$ are owned by Player $i$ ($V_0\cap V_1=\emptyset$).

In both types of games, a token is placed on the initial vertex $v_0$. In each turn of the game, the owner of the current vertex of the token must choose an outgoing edge of that vertex. The token is then sent along that edge to the next vertex, and then the next turn starts. The game continues infinitely, and assumed is that every vertex has at least one outgoing edge so the token does not get stuck. 

The game graph additionally has edge weights $w:E\to\mathbb{Z}$. Suppose that in the $i$-th turn the token is sent along edge $e_i$ for all $i\in \mathbb{N}$. The winner of the game is then determined by the sequence of weights encountered on the infinite walk the token travels along, which is the sequence $w(e_1),w(e_2),w(e_3),\ldots$.

For parity games, the winner is determined by the parity of the largest weight that is seen infinitely often:
\[
\limsup_{i\to\infty}(w(e_i)) \mod 2
\]
Player 0 wins if that largest weight is even, and Player 1 wins if it is odd. While parity games are usually considered with weights (often called priorities) on the nodes instead of the edges, a formulation with node weights can easily be transformed into a formulation with edge weights and vice versa.

For mean payoff games, the players have optimization objectives. Specifically, Player 0 tries to maximize the average weight encountered:
\[
\limsup_{t\to\infty}\frac{1}{t}\sum_{i=1}^tw(e_i)
\]
We also refer to Player 0 as the Maximizer. Player 1 (or: the Minimizer) tries to minimize the average weight. Likewise, for these games we denote $\vm:=V_0$ and $\vn:=V_1$.

It is well-known that both parity games and mean payoff games have a well-defined value of the game and optimal strategies. Moreover, there always exist \emph{positional} optimal strategies: where positional means that a player always makes the same choice when the token is on the same vertex. Hence we can consider Player 0 strategies as functions $\sigma:V_0\to V$, where $\sigma(v)$ signifies that Player 0 picks edge $(v,\sigma(v))$ whenever the token lands on $v\in V_0$. For convenience of notation, we also consider $\sigma$ as the set of edges used by a positional Player 0 strategy (e.g. we write both $\sigma(v)=v'$ and $(v,v')\in\sigma$). Similarly, we consider Player 1 strategies to be functions $\tau:V_1\to V$. We denote the set of edges Player 0 can use by $E_0:=\SetOf{(v,v')}{v\in V_0}$, and similarly we have $E_1:=\SetOf{(v,v')}{v\in V_1}$. If Player 0 sticks to strategy $\sigma$, this means that the token can only move along the edges of the subgraph $G_{\sigma}:=(V_0\cup V_1,E_{\sigma})$, where $E_{\sigma}=\SetOf{(v,\sigma(v))}{v\in V_0}\cup E_1$. And if Player 1 uses strategy $\tau$, then the token can only move in $G_{\tau}:=(V_0\cup V_1,E_{\tau})$, with $E_{\tau}:=\SetOf{(v,\tau(v))}{v\in V_1}\cup E_0$.

\paragraph{Sink parity games} are parity games that additionally fulfill the following conditions: 
\begin{enumerate}
    \item There exists a so-called \emph{sink node} $\top$, whose only outgoing edge is $(\top,\top)$, which has a priority of $-\infty$.
    \item There exists a Player 0 strategy $\sigma$ such that the highest priority of any cycle in $G_{\sigma}$ (except the 1-vertex cycle $\top$) is even.
    \item There exists a Player 1 strategy $\tau$ such that the highest priority of any cycle in $G_{\tau}$ (except the 1-vertex cycle $\top$) is odd.
\end{enumerate}
We call the strategies for Player 0 and Player 1 that fulfill the above conditions \emph{admissible}. These games can be thought of in some sense as a draw, since no player can do better than sending the token to the sink node. In this case, we consider alternative winning conditions. Here Player 0 wants to maximize the following quantity, which Player 1 tries to minimize.
\begin{equation}\label{eq:sinkparobj}
    \sum_{i=0}^{\infty}(-t)^{w(v_i)}
\end{equation}
Where $t$ is some large positive constant. This means that Player 0 tries to meet large even priorities before entering the sink, and avoid large odd priorities. See \cite{Friedmann2011ExponentialPrograms} for more details on sink parity games.

\paragraph{Longest shortest path problems} or LSPs are another type of games on graphs, introduced in \cite{bjorklund_combinatorial_2007}. We use a slight variant which can also be seen as a total payoff game \cite{brihaye_pseudopolynomial_2017}. An LSP is played like a mean payoff game, where the weighted graphs fulfill some additional conditions:
\begin{enumerate}
    \item There exists a so-called \emph{sink node} $\top$, whose only outgoing edge is $(\top,\top)$, which has a weight of $0$. We assume $\top\notin \vn,\vm$.
    \item There exists a Maximizer strategy $\sigma$ such that the total weight of any cycle in $G_{\sigma}$ (except the 1-vertex cycle $\top$) is positive.
    \item There exists a Minimizer strategy $\tau$ such that the total weight of any cycle in $G_{\tau}$ (except the 1-vertex cycle $\top$) is negative.
\end{enumerate}
We call strategies that satisfy the second or third condition \emph{admissible}. If we would consider this game as a mean payoff game, the best that players can do is to send the token to the sink, similar to sink parity games. Therefore we consider an alternative objective for the players, which is to maximize/minimize the total weight:
\[
\sum_{i=1}^{\infty}w(e_i)
\]
By the three conditions above, the total weight will be finite, given optimal play from both players. 

It is well-known that the problem of determining the vertices in an MPG with value $\geq 0$ reduces to solving an LSP \cite{bjorklund_combinatorial_2007,brihaye_pseudopolynomial_2017}. Moreover, there is a reduction from parity games to MPGs (simply replacing $w(e)$ by $(-|V|)^{w(e)}$ for each edge), and it is well-known that one can solve parity games by reducing them to a sink parity game (see e.g. \cite[Lem. 4]{van_dijk_worst-case_2024} and references therein). A more formal introduction to these games, proofs of their properties, and algorithms can be found in \cite{fijalkow_games_2023}.

\paragraph{Strategy improvement} is a technique for finding optimal strategies that can be applied to many types of games. We are using a variant of strategy improvement for LSPs first introduced in \cite{bjorklund_combinatorial_2007}. It is the LSP equivalent of the strategy improvement algorithms from \cite{voge_discrete_2000,puri_theory_1997}. It works as follows:
\begin{enumerate}
    \item Start with some admissible Player 0 strategy $\sigma$
    \item For each node $v$, compute $\val_{\sigma}(v)$, which is the smallest weight of any $v$,$\top$-path in $G_{\sigma}$.
    \item Mark all edges $(v,v')\in E_0$ with $w(v,v')+\val_{\sigma}(v')>\val_{\sigma}(v)$ as \emph{improving}.
    \item If there are no improving edges, return $\sigma$
    \item Otherwise, for one improving edge $(v,v')$, change $\sigma(v)$ to $v'$, and go back to 2.
\end{enumerate}
 We call step 5 an improving switch. There is some freedom to choose the edge that is used for the switch. It is even possible to make multiple improving switches at once in step 5, instead of just one switch: switching any combination of improving edges simultaneously yields an improved strategy. The choice which edge(s) to switch is made by some so-called \emph{improvement rule}. For our purposes, we only consider improvement rules that allow for one switch at a time: this is connected to the fact that the simplex algorithm only exchanges one element of the basis at a time. \\
 Note: in LSP and sink parity games, many variants of strategy improvement are equivalent to the version we use. However, one should be careful if one reduces a game to LSP. For example, a run of the algorithm from \cite{voge_discrete_2000} in a (non-sink) parity game will change in general if the game is turned into a sink parity game.

 \paragraph{(average reward) Markov decision processes} or MDPs can be seen as a mean payoff game where the Minimizer is replaced by a Randomizer: at each Randomizer node, the token moves to a random successor according to some fixed probability distribution. The goal of the Maximizer is the same as for MPGs: maximize the expected average reward. For more details we refer to \cite{puterman1994markov}. A special case of MDPs is weak unichain MDPs: these can be seen as LSPs where the Minimizer is replaced by a Randomizer. Admissible strategies are given by Maximizer strategies that are guaranteed to reach the sink eventually - and we require of a weak unichain MDP that its optimal policy is admissible. For more details on these we refer to \cite{Friedmann2011ExponentialPrograms}.

 \section{From LSP to LP}\label{sec:LSPtoLP}
In this section, we prove the main result, which relates strategy improvement in LSPs to the simplex algorithm. From here, we assume that the LSP is nondegenerate, defined as follows.
\begin{definition}
    We call an LSP \emph{nondegenerate} if it satisfies two conditions:
    \begin{itemize}
        \item In the underlying weighted graph, there is no cycle of total weight equal to 0, except the cycle with $\top$ (also called simple arenas, see e.g. \cite{ohlmann_gkk_2022}).
        \item For every pair of distinct nodes $v,v'$, and for every pair of $v,v'$-paths $P_1,P_2$ that share no internal vertices, we have $w(P_1)\neq w(P_2)$.
    \end{itemize}
\end{definition}
This is a more relaxed condition\footnote{Tropical genericity of the tropical LP means there are no two distinct matchings $M_1,M_2$ in the underlying undirected graph that are incident to the same set of vertices and have the same weight, if the weights of edges in $\emn$ are multiplied by $-1$. Tropical genericity implies our two conditions: no zero-weight cycles (the Maximizer resp. Minimizer edges of the cycle would form a pair of forbidden matchings), and no vertex-disjoint $v,v'$-paths $P_1,P_2$ of equal weight (then the matchings $(\emx\cap P_1)\cup (\emn\cap P_2)$ and $(\emx\cap P_2)\cup (\emn\cap P_1)$ are forbidden).} compared to the condition of tropical genericity for bipartite MPGs, which is used in e.g. \cite{allamigeon_combinatorial_2014}.

To express the terms of the LP we construct, we need to define two functions, which are called $\minwalks$ and $\walks$, respectively. For two vertices $i,j\in V$ and $m\in\mathbb{Z}$, let $a_m(i,j)$ be the number of finite $i,j$-walks in the graph $(V, \emn)$, such that the path has weight exactly $m$. For general graphs, this number may not be finite. However, we know there cannot be a negative weight cycle in $\emn$: otherwise the Minimizer can create a negative cycle for every maximizer strategy, which contradicts the second assumption of LSP problems. There can also not be a 0-weight cycle in $\emn$, by nondegeneracy and since $(\top,\top)\notin \emn$. Thus there is only a limited number of cycles that can be contained in a walk of weight $m$, limiting the number of edges such a walk can contain.

With this, we define the series
\[
 \minwalks_{ij}(t)=\sum_{m\in \mathbb{Z}}a_m(i,j)t^{-m} \enspace . 
 \]
This polynomial possibly has infinitely many terms, but if $t$ is large enough, we can show that the sum converges: this also means that there is a leading coefficient, e.g. the largest $m$ with $a_m(i,j)\neq 0$. Thus, $\minwalks$ would be a Puiseux series if $t$ was replaced by $\frac{1}{t}$. However, for our purposes, we will consider $t$ simply as a large enough positive integer.

Before proving convergence, we first need the following lemma:
\begin{lemma}[generalized ratio test]\label[lemma]{lem:genratio}
    Let $s\in \mathbb{N}$, and let $0<\alpha<1$. Let $(a_n)_{n\in \mathbb{N}}$ be a sequence of nonnegative real numbers, with the property $a_n\leq \frac{\alpha}{s}\sum_{i=n-s}^{n-1}a_i$ for all $n>s$. Then the infinite series $\sum_{i=1}^{\infty}a_i$ converges.
\end{lemma}
\begin{proof}
    Consider the sequence $(b_n)_{n\in\mathbb{N}}$, defined recursively by
    \[
    b_n=\begin{cases}
        a_n & n\leq s\\
        \frac{\alpha}{s}\sum_{i=n-s}^{n-1}b_i & n>s
    \end{cases}
    \]
    Clearly $b_n\geq a_n$ for all $n$, so we just need to show that $\sum_{i=1}^{\infty}b_i$ converges. For all $n>s$ we have $b_n\leq \alpha \max \{b_{n-s},b_{n-s+1}, \ldots, b_{n-1}\}$. If $M=\max\{b_i:i\leq s\}$, then it immediately follows that $b_i\leq \alpha M$ for $i=s+1,s+2, \ldots, 2s$. From that we can see that $b_i\leq \alpha^2M$ for $i=2s+1, 2s+2, \ldots, 3s$, and continuing this argument inductively we get $b_i\leq \alpha^{\lceil \frac{i}{s}\rceil-1}M$. But since
    \[
    \sum_{i=1}^{\infty}\alpha^{\lceil \frac{i}{s}\rceil-1}M = \frac{Ms}{1-\alpha}
    \]
    it follows that $\sum_{i=1}^{\infty}b_i$ converges, and from that follows the lemma.
\end{proof}

Now we are ready to prove convergence for $\minwalks$.
\begin{lemma}\label[lemma]{lem:convergingseries}
    Given an LSP problem and pair of nodes $i,j$, there exists a $T\in \mathbb{R}$ such that $\minwalks_{ij}(t)$ is finite for all $t\geq T$.
\end{lemma}
\begin{proof}
     Let $G'=(V,\emn)$, then we know that the weight of any cycle in $G'$ must be positive. It follows that there is a lower bound on the weight of all walks in $G'$, hence there is some $M\in \mathbb{Z}$ such that $a_m(i,j)=0$ for all $m<M$. Let $N$ be the maximum weight of any path in $G'$. Let $S$ be the maximum weight of any cycle in $G'$. Let $C$ be the number of cycles in $G'$. We show that if $m>N$, then the number of walks of length $m$ is at most $SC$ times the number of walks of length between $m-S$ and $m-1$. 
     
     To show this, let $P$ be a walk in $G'$ of weight exactly $m$. If $m>N$, $P$ cannot be a path. Let $v$ be the first node that we encounter twice in $P$. We can split the edges of $P$ into two walks: one cycle $c$ from the first to the second occurrence of $v$, and the remainder $Q$ when $c$ is removed from the middle of $P$. Since $c$ has positive weight of at most $S$, we know $m-S\leq w(Q)\leq m-1$. The walk $Q$ is uniquely determined from $P$. On the other side, there are at most $|V|C$ walks that can be obtained by inserting a cycle into $Q$, given that the cycle is inserted at the first occurrence of its starting node. If $R$ is the number of walks of length between $m-s$ and $m-1$, it follows that there can be at most $|V|CR$ walks of weight exactly $m$. Phrased differently, we get $a_m(i,j)\leq |V|C\sum_{k=m-S}^{m-1}a_{k}(i,j)$.
     Now pick $T>|V|CS$. For any $t\geq T$ we get
     \begin{eqnarray*}
          a_m(i,j)t^{-m}&\leq & |V|C\sum_{k=m-S}^{m-1}a_{k}(i,j)t^{-m}\\
          &\leq & \frac{|V|C}{t}\sum_{k=m-S}^{m-1}a_{k}(i,j)t^{-k}\\
          &\leq & \frac{|V|C}{T}\sum_{k=m-S}^{m-1}a_{k}(i,j)t^{-k}
     \end{eqnarray*}
     Since $\frac{|V|C}{T}<\frac{1}{S}$, it follows from Lemma \ref{lem:genratio} that the series $\minwalks_{ij}(t)$ converges for all $t\geq T$. 
\end{proof}

Now suppose we have some admissible Maximizer strategy $\sigma$. Let $\walks_{ij}^{\sigma}(t)$ be the polynomial defined by $\sum_{m\in \mathbb{Z}}a_m^{\sigma}(i,j)t^{-m}$, where $a_m^{\sigma}(i,j)$ is the number of $i,j$-walks in $G_{\sigma}$ of weight exactly $m$. Finiteness of $a_{m}^{\sigma}(i,j)$ and convergence of this series can be shown analogous to Lemma \ref{lem:convergingseries}. Again, using the fact that $G_{\sigma}$ has only positive cycles (outside of $\top$). Note that, for any node $i$, we have $\walks^{\sigma}_{i,i}(t)\geq 1$, since the trivial path has weight $0$. We also have $\walks^{\sigma}_{\top\top}(t):=1$, since $(\top,\top)\notin G_{\sigma}$.

Now we are ready to formulate the LP that we use for our main result. We claim that the below LP is a formulation of the LSP problem.
\begin{eqnarray}
\min & z &\nonumber\\
s.t. & A\begin{bmatrix}
x\\
z
\end{bmatrix} &=b\label{eq:LSPLP}\\
&\begin{bmatrix}
x\\
z
\end{bmatrix}  &\geq 0\nonumber\\
\text{where} & A_{ij} &=\begin{cases} 
1-t^{-w(j)}\minwalks_{i',i}(t) & j=(i,i')\in \emx\\
-t^{-w(j)}\minwalks_{j_2,i}(t) & j=(j_1,j_2)\in \emx,j_1\neq i\\
1 & (j,i)=(z,\top)\\
0 & j=z, i\neq \top
\end{cases}\nonumber
\end{eqnarray}
$A$ has one row for every vertex $i\in \vm$, one column for every $j\in \emx$ (related to variable $x_j$), and one column for $\top$ (related to the variable $z$).

The vector $b$ is is constant. Its exact value does not matter for our reduction, as long as its values are positive. In some cases, some values of $b$ could be zero, for example $b_{\top}=0$ would also work. But for the sake of simplicity we always assume $b>0$.

The LP has one equation for every Player 0 node $i$, and one equation for $\top$. In the equation related to $i$, there is a positive term for each outgoing edge $(i,i')$, and a negative term for every edge from which you can reach $i$ via Minimizer nodes. 

\paragraph{Connection to LP formulation of MDP.}
One may view the LP intuitively as the LSP version of the flow formulation of MDPs as in \cite{Friedmann2011ExponentialPrograms,puterman1994markov}: if we interpret our LP as a flow problem, every Maximizer node $i$ produces some flow $b_i$ and receives some incoming flow. The flow is then sent along the outgoing edge(s) that the Maximizer wants to use. Every edge has a cost per unit flow, which the Maximizer tries to maximize. The LP is constructed such that a flow $f$ on an edge results in the related variable being close to $t^{-f}$. The maximization and minimization objectives of the players in the LSP are then simulated in different ways. The Maximizer is simulated by the minimization objective of the LP. In the LP, when the variable of an edge with weight $w(e)$ is used, the corresponding variable is multiplied by $t^{-w(e)}$. We let the Maximizer try to maximize $-\log_{t}(z)$, where $z$ represents the flow towards the sink, which is conceptually the same as giving a reward of $w(e)$ when an edge $e$ is used. For a Minimizer node, taking the minimum is simulated by copying the total inflow $t^{-f}$ to \emph{all} of its outgoing edges. This works since the other flows become negligible compared to the flow of the shortest paths in the LSP: if $S$ is the set of weights of all possible walks that the Minimizer can use to reach the sink, we have $t^{-\min(S)}\approx\sum_{s\in S}t^{-s}$, because $t$ is large. A similar idea is also used in \cite{schewe_parity_2009}, which explains this concept in more detail.

\paragraph{Proving the main theorem.} 
The rest of this section is dedicated to formally proving that \eqref{eq:LSPLP} is an LP formulation of LSPs, with a connection between strategy improvement and the simplex algorithm. We compare $\minwalks$ and $\walks$ by their value as $t\to \infty$, so $\minwalks_{ij}>\minwalks_{kl}$ if there is a $T$ such that $\minwalks_{ij}(t)-\minwalks_{kl}(t)>0$ for all $t>T$. This is in principle the same as comparing the sorted sequences of possible lengths of walks lexicographically: the shortest walk is the most important, and if the shortest walks are equal, we compare the second shortest walk, third shortest, etc. Alternatively, one could view the LP as a linear program over the field of Puiseux series, similar to \cite{allamigeon2015tropicalizing}.

The main result for this section is the following theorem, which relates the key concepts of strategy improvement to the key concepts of the simplex algorithm.

\begin{theorem}\label{thm:main}
Suppose the LP \eqref{eq:LSPLP} is derived from a nondegenerate LSP, and let $\B$ be some basis of the LP. Then 
    \begin{enumerate}
        \item Suppose $\B$ is the related basis to an admissible Maximizer strategy $\sigma$, i.e. $z\in \B$ and for the other variables $x_j\in \B\Leftrightarrow j\in \sigma$. In this case, every column of $B$ corresponds to an edge, each with a unique tail. Thus we can index the columns of $B$ by the tails of the corresponding edges, index the column corresponding to $z$ by $\top$, and the rows of $B$ by the corresponding vertices. Then $B^{-1}_{i_1i_2}=\walks^{\sigma}_{i_2i_1}(t)$ for all $i_1,i_2\in \vm\cup\{\top\}$.
        \item The feasible region of the LP forms a simple polyhedron, and basis $\B$ is feasible if and only if the basic variables consist of $z$ and of $x_j$ for all $j\in \sigma$, for some admissible Player 0 strategy~$\sigma$.
    \item $H_{1,e}<0$ if and only if $e$ is an improving move for strategy improvement, and the set of optimal solutions of the LP consists of a single vertex.
    \end{enumerate}
It follows that for every possible run of strategy improvement in the LSP, there is an equivalent run of the simplex algorithm in the LP (\ref{eq:LSPLP}), and vice versa.
\end{theorem}

\begin{proof}

 \emph{Item 1.} Let $\tilde{B}=I-B$. We can index the rows and columns of $B$ as indicated in the theorem, simply because a Maximizer strategy by definition has exactly one outgoing edge for each vertex in $\vm$. Assume some order of $\vm \cup\{\top\}$. Suppose the columns of $B$ are arranged in order of the tails of the corresponding edges, and the equations are arranged by their corresponding vertex. We assume the column corresponding to $z$ is last, and thus the row corresponding to $\top$ is the bottom one. By definition of $B$ we get
\begin{equation*}
\tilde{B}_{ij} =\begin{cases} 
0 & j=z\\
t^{-w(j, \sigma(j))}\minwalks_{\sigma(j),i}(t) & j\neq z
\end{cases}
\end{equation*}
Now we want to show that the inverse of $B$ is given by the infinite series $I+\tilde{B}+\tilde{B}^2+\ldots$. Let $\kw{k}^{\sigma}_{ij}$ be a polynomial defined by $\sum_{m\in \mathbb{Z}}a_{m}^{\sigma,k}(i,j)t^{-m}$, where $a_{m}^{\sigma,k}(i,j)$ is the number of finite $i,j$-walks in $G_{\sigma}$ of weight exactly $m$, with the restriction that exactly $k$ nodes of the walk are in $\vm$ (we can count the same node multiple times here, and we do not count the end node of the walk). From the definition we get $\kw{0}^{\sigma}_{ji}(t)=\minwalks_{ji}(t)$. If $j\in \vm$, then every walk in $\kw{1}^{\sigma}_{ji}(t)$ starts with edge $(j, \sigma(j))$, hence $\kw{1}^{\sigma}_{ji}(t)=t^{-w(j,\sigma(j))}\minwalks_{\sigma(j)i}(t)=\tilde{B}_{ij}$. Likewise, for $j=z$ we have $\kw{1}^{\sigma}_{ji}(t)=0=\tilde{B}_{ij}$ as $\top\notin \vm$. Now we prove by induction that 
\begin{equation}\label{eq:Btildek}
\tilde{B}_{ij}^k = 
\begin{cases} 
0 & j=z\\
\kw{k}^{\sigma}_{ji}(t) & \text{otherwise}\\
\end{cases}
\end{equation}
The induction basis follows from our observation about $\kw{1}$. For the induction step, the key insight is that every path with $k$ nodes of $\vm$ consists of one path with one $\vm$-node and one path with $k-1$ nodes from $\vm$. For the induction step, we assume that (\ref{eq:Btildek}) holds for $k-1$, and then we get (if $j\neq z$):
\begin{eqnarray*}
\tilde{B}^k_{ij} &=& \sum_{v\in \vm\cup \{\top\}}\tilde{B}_{iv}\tilde{B}^{k-1}_{vj}\\
&=&\sum_{v\in \vm\cup \{\top\}}t^{-w(v,\sigma(v))}\minwalks_{\sigma(v)i}(t) \times (k-1)\operatorname{-walks}^{\sigma}_{jv}(t)\\
&=& k\operatorname{-walks}^{\sigma}_{ji}(t)
\end{eqnarray*}
and if $j=z$ we get likewise $\tilde{B}^k_{ij}=0$ from the induction hypothesis. This completes the induction proof of (\ref{eq:Btildek}). From this it immediately follows that 
\[
(I+\tilde{B}+\tilde{B}^2+\ldots)_{ij} = \walks^{\sigma}_{ji}(t) \; \forall i,j \in \vm\cup\{\top\}
\]
where we also use that $\walks^{\sigma}_{\top\top}(t)=1$ and $\walks^{\sigma}_{\top i}(t)=0$ for all $i\in\vm$. Since we know that $\walks_{ji}^{\sigma}(t)$ converges if $\sigma$ is admissible and $t$ is large enough, we conclude that $I+\tilde{B}+\tilde{B}^2+\ldots$ converges elementwise and is therefore well-defined. It also follows that $\tilde{B}^k\to 0$ elementwise as $k\to \infty$. Working out the brackets of $(I-\tilde{B})(I+\tilde{B}+\tilde{B}^2+\ldots)$ then gives a telescoping sum that is equal to $I$. Recalling that $I-\tilde{B}=B$, we conclude $B^{-1}=I+\tilde{B}+\tilde{B}^2+\ldots$, so $B^{-1}_{i_1i_2}=\walks^{\sigma}_{i_2i_1}(t)$ for all $i_1,i_2\in \vm\cup\{\top\}$.

\emph{Item 2.} First of all, note that the variable $z$ must be in any feasible basis because $z$ is the only variable with positive coefficient in the $\top$-row of the LP (\ref{eq:LSPLP}). We treat all the possible options case by case. 
\begin{itemize}
    \item Suppose $\B$ does not correspond to a proper strategy. That means there is a vertex $i\in \vm$ such that there is no column in $B$ that corresponds to an outgoing edge of $i$. In particular, this means that all entries in the $i$-th row of $B$ are 0 or of the form $-t^{-w(j_1, j_2)}\minwalks_{j_2,i}(t)$, in particular they are all nonpositive. Looking at the right side of the $i$-th equation, we have $b_i >0$ by assumption. That is not possible, since then $i$-th equality then dictates that at least one element of $x_B$ must be negative, so $x$ is not feasible. We conclude that any feasible basis must correspond to a Maximizer strategy.
    \item Suppose $\B$ corresponds to a non-admissible strategy $\sigma$, and suppose that $\B$ is feasible. This means that there is a cycle $C$ in $G_{\sigma}$ with negative weight. Let $v_1,v_2, \ldots, v_k$ be the vertices in $\vm\cap C$ (in that order on the cycle). For simplicity of notation we say $v_{k+1}:=v_1$. We divide the cycle into segments $P_1, P_2,\ldots, P_k$, where every segment $P_i$ starts with $v_i$ and ends at $v_{i+1}$. By definition of $\minwalks$, we have
\[
t^{-w(P_i)}\leq t^{-w(v_i, \sigma(v_i))}\minwalks_{\sigma(v_i),v_{i+1}} \; i=1, \ldots , k
\]
Now consider the $i$-row in the system of equations $Bx_{B}=b$. Since $b_i>0$, and since the element on the diagonal of $B$ is the only positive element of the $i$-row, we must have $(x_B)_i>0$ if we want $x_B\geq 0$. Since we already found that any feasible basis must be a strategy, this implies that the polyhedron is simple, since we cannot pick another basis for the same basic solution. In particular, $(x_B)_{v_{1}},\ldots,(x_B)_{v_{k}}>0$. That leads to the following inequalities, where the equality $*$ is equivalent to the equation in the $v_i$-row of $Bx_{B}=b$:
\begin{eqnarray*}
x_{v_i} &\stackrel{*}{=}& b_i+\sum_{j\in \vm}t^{-w(j, \sigma(j))}\minwalks_{\sigma(j), v_i}(t) (x_B)_j\\
&>& t^{-w(v_{i-1}, \sigma(v_{i-1}))}\minwalks_{\sigma(v_{i-1}), v_i}(t)(x_B)_{v_{i-1}} \\
&\geq& t^{-w(P_{i-1})}(x_B)_{v_{i-1}} \;\; \text{for } i=2,3,\ldots, k+1
\end{eqnarray*}
Applying these inequalities for $i=2,3,\ldots, k+1$ simultaneously gives
\[
(x_B)_{v_{1}} > t^{-\sum_{i=1}^kw(v_i, \sigma(v_i))}(x_B)_{v_{1}}=t^{-w(C)}(x_B)_{v_{1}}>(x_B)_{v_{1}}
\]
as $w(C)<0$ and $(x_B)_{v_{1}}>0$. But that is not possible, so we conclude that a feasible basis must come from an admissible maximizer strategy. 
\item Finally, suppose we have a basis coming from an admissible strategy. From item 1 we know the values in $B^{-1}$. In particular, all elements of $B^{-1}$ are nonnegative. We have $x_B=B^{-1}b$ and $b>0$, so clearly $x_B\geq 0$ and the basis is feasible.\footnote{The difference with a non-admissible strategy here lies in the fact that for a non-admissible strategy the value of the polynomial $\walks^{\sigma}_{ij}(t)$ is not defined for all $i,j$: the infinite series does not converge if the Minimizer can reach a negative-weight cycle in $G_{\sigma}$. It also follows that then the series $I+\tilde{B}+\tilde{B}^2+\ldots$ does not converge, so it will not be equal to $B^{-1}$.}
\end{itemize} 

\emph{Item 3.} Let $e=(j_1,j_2)$ be an edge that is not in the basis $\B$ as an index. We first compute the reduced cost $H_{1,e}$, which is defined as $H_{1,e}=c_e-c_{B}^TB^{-1}A_e$. Here $A_e$ is the $e$-column of $A$. Note that $c_e=0$ and $c_B^T=(0,0,\ldots, 0, 1)$, since only $z$ occurs in the objective. Filling this into the definition, we get 
\begin{eqnarray}
 H_{1,e}&=&-\sum_{v\in \vm\cup\{\top\}}B^{-1}_{\top v}(A_e)_v\nonumber\\
 &=&-B^{-1}_{\top z}A_{\top e}-B^{-1}_{\top j_1}A_{j_1e}-\sum_{v\in \vm\backslash \{j_1\}} B^{-1}_{\top v}A_{ve} \nonumber \\
 &=& 1\cdot t^{-w(e)}\minwalks_{j_2,\top}(t) -\walks^{\sigma}_{j_1\top}(t)\cdot (1-t^{-w(e)}\minwalks_{j_2,j_1}(t))\nonumber\\
 & &- \sum_{v\in \vm\backslash\{j_1\}} \walks^{\sigma}_{v\top}(t)\cdot -t^{-w(e)}\minwalks_{j_2,v}(t)\nonumber\\
 &=& t^{-w(e)}\left(\minwalks_{j_2,\top}(t)+\sum_{v\in \vm} \walks^{\sigma}_{v\top}(t)\cdot\minwalks_{j_2,v}(t)\right)\nonumber\\
 & &-\walks^{\sigma}_{j_1\top}(t)\label{eq:reducedcost}
 \end{eqnarray}

The last equality (\ref{eq:reducedcost}) is just the result of reordering the terms. Now we want to rewrite the last line. To do so, we look at $\walks^{\sigma}_{j_2,\top}(t)$. Any path that leads from $j_2$ to $\top$ in $G_{\sigma}$ falls in one of two categories:
\begin{itemize}
    \item It goes immediately to $\top$ without passing any maximizer nodes (this also means $j_2\in \vn$).
    \item It passes (or starts at) a node from $\vm$.
\end{itemize}
This implies that
\begin{equation}\label{eq:rewritewalks}
\walks^{\sigma}_{j_2\top}(t)= \minwalks_{j_2,\top}(t) +\sum_{v\in \vm} \walks^{\sigma}_{v\top}(t)\minwalks_{j_2,v}(t) 
\end{equation}
Now if we substitute this in (\ref{eq:reducedcost}), we get
 
 \begin{eqnarray}
 H_{1,e} &=& t^{-w(e)}\walks^{\sigma}_{j_2\top}(t) - \walks^{\sigma}_{j_1\top}(t)\label{eq:reducedcost2}
\end{eqnarray}
Now we want to prove that $H_{1,e}<0$ if and only if $e$ is an improving move for strategy improvement. 
From equation (\ref{eq:reducedcost2}) we know that $H_{1,e}<0$ is equivalent to saying $t^{-w(e)}\walks^{\sigma}_{j_2\top}(t) < \walks^{\sigma}_{j_1\top}(t)$. Let $P_1$ and $P_2$ be the shortest paths in $G_{\sigma}$ from respectively $j_1$ to $\top$ and from $j_2$ to $\top$. Since $\sigma$ is admissible, we know that shortest paths exist. The leading term of $\walks^{\sigma}_{j_1\top}(t)$ has exponent $-w(P_1)$, while that of $t^{-w(e)}\walks^{\sigma}_{j_2\top}(t)$ has as exponent $-w(e)-w(P_2)$. We claim that these two exponents cannot be the same. 

If $P_2$ passes through $j_1$, then $e,P_2$ is a walk that consists of a cycle (which we call $C$) combined with a shortest $j_1,\top$-path $P_3$ in $G_{\sigma}$ (where being a shortest path implies $w(P_3)=w(P_1)$). But then $-w(e)-w(P_2)=-w(P_1)$ implies $w(C)=0$, which contradicts the definition of nondegenerate LSPs. On the other hand, if $P_2$ does not pass through $j_1$, then $e,P_2$ is a path in $G$ that is not equal to $P_1$. In that case, let $j_3$ be the first vertex on $P_1$ that is not $j_1$ and is also on $P_2$ ($j_3$ exists, since both paths end at $\top$). Let $P_1'$, $P_2'$ be the sections from $j_1$ to $j_3$ of $P_1$ and  $e,P_2$, respectively. Then we can prove that no internal vertex $v$ of $P_2'$ is in $P_1$ by contradiction.

Suppose there is such a vertex $v$. Then there is a $v,j_3$-path $P_2''$ inside $P_2$ of weight $\operatorname{distance}(j_3,\top)-\operatorname{distance}(v,\top)$. By definition of $j_3$, $v$ must lie beyond $j_3$ on $P_1$, hence there is a $j_3,v$-path $P_1''$ inside $P_1$ of weight $\operatorname{distance}(j_3,\top)-\operatorname{distance}(v,\top)$. The paths $P_1''$ and $P_2''$ together form a $0$-weight closed walk in $G_{\sigma}$. Any closed walk of weight 0 contains a cycle of weight $\leq 0$. But there are no 0-weight cycles in $G\backslash\{\top\}$ due to nondegeneracy, and no negative-weight cycles in $G_{\sigma}$ due to admissibility of $\sigma$. Hence, such a $v$ does not exist. But then $P_1'$ and $P_2'$ share no internal vertices, and since they are both shortest $j_1,j_3$-paths, that means $w(P_1')=w(P_2')$. Thus, we contradict nondegeneracy, and that completes the proof that $-w(e)-w(P_2)\neq -w(P_1)$.
Hence the two exponents that we mentioned must always be different. In particular, this means that $H_{1,e}<0$ if and only if the exponent of the leading term of $\walks^{\sigma}_{j_1\top}(t)$ is larger than the other term, if and only if $w(P_2)+w(e)>w(P_1)$. This last statement in turn can be written in terms of the valuation (recall from `strategy improvement' in \cref{sec:preliminaries}), in which case we get $\val_{\sigma}(j_2)+w(e) > \val_{\sigma}(j_1)$. This is the same as saying $e$ is an improving edge in strategy improvement. So indeed, $H_{1,e}<0$ if and only if $e$ is an improving move for $\sigma$ in strategy improvement. 

Since there exists an admissible strategy for every LSP, it follows that the LP \eqref{eq:LSPLP} has a BFS. Since the objective value of the LP is bounded from below (as $z\geq 0$), this means there is a BFS in the LP that forms an optimal solution (by correctness of the simplex method). If we consider one optimal BFS and some non-basic index $e$, then we just found $H_{1,e}=0$ is not possible. Hence $H_{1,e}>0$ for all basic variables, and that implies that this BFS is the only optimal solution.

\emph{Conclusion.} Now we compare runs of the simplex algorithm with runs of strategy improvement. Suppose that some basis $\B$ corresponds to an admissible strategy $\sigma$, and suppose variable $x_e$ is improving in the LP, with $e=(j_1,j_2)$. The objective value of the LP is bounded from below (as $z\geq 0$), hence pivoting on $e$, yields a new feasible basis. We know that feasible bases correspond to proper, admissible strategies. Therefore only $(j_1,\sigma(j_1))$ is allowed to leave the basis to maintain feasibility, as the new basis must correspond to a proper strategy. Hence the next basis has to correspond to the strategy $\sigma[e]$, defined by
\[
\sigma[e](v)=\begin{cases}
    \sigma(v) & v\neq j_1\\
    j_2 & v=j_1
\end{cases}
\]
And that is exactly the strategy we gain by choosing the improving edge $e$ in strategy improvement. Likewise, every improving move $e$ corresponds to an improving variable $x_e$ in the LP, which then leads to the basis related to the new LSP strategy. 

Moreover, both algorithms terminate at the same point. Because of the correctness of the strategy improvement algorithm, we know that there exists a unique admissible strategy $\sigma^*$ that is optimal, which is reached when strategy improvement terminates. From this strategy, there are no improving moves possible. This in turn means that in the basis $\mathcal{B}$ corresponding to $\sigma^*$, there are no improving indices, since all reduced costs are nonnegative. As a result, by correctness of the simplex algorithm, the basis $\mathcal{B}$ is optimal for the LP, and its related BFS is an optimal solution for the LP \eqref{eq:LSPLP}. So indeed, every run of the two algorithms relates to a run of the other. This completes the proof of \cref{thm:main}.
\end{proof}

In conclusion, \cref{thm:main} tells us that there is an LP formulation of any nondegenerate LSP problem, in which the simplex algorithm behaves like strategy improvement. Of course, if we want to analyze the behavior of a certain pivot rule for the simplex algorithm, we would need to find an equivalent improvement rule for strategy improvement.

In support of this, we can find some interpretations of variables from the simplex tableau. This is done by carefully looking at the proof of \cref{thm:main}, and is summarized in the following theorem.
\begin{theorem}\label{thm:interpretation} Consider the feasible basis $\mathcal{B}$ for \eqref{eq:LSPLP} and its related strategy $\sigma$ as in \cref{thm:main}. The following hold:
\begin{enumerate}
\setcounter{enumi}{3}
    \item If $e=(j_1,j_2)\in \sigma$, $\mathcal{B}$ is the related basis to $\sigma$, and $x$ the basic solution of $\mathcal{B}$, then 
    \[
        x_e=\sum_{v\in \vm}b_v\walks^{\sigma}_{v,j_1}(t) \text{  and  } z=\sum_{v\in \vm\cup\{\top\}}b_v\walks^{\sigma}_{v,\top}(t)
    \]
    \item If $(j_1,j_2)\notin \mathcal{B}$ then for all $i\in \vm\cup\{\top\}$ we have
    \[H_{i,(j_1,j_2)}=t^{-w(j_1,j_2)}\walks^{\sigma}_{j_2i}(t) - \walks^{\sigma}_{j_1i}(t)\]
\end{enumerate}
\end{theorem}
\begin{proof}
    \emph{Item 4.} This follows immediately from the equation $x_B=B^{-1}b$ and the values of $B_{ij}^{-1}$ we found in the proof of item 1.

\emph{Item 5.} Let $e=(j_1,j_2)$. We have by definition $H_{i, (j_1,j_2)}=(-B^{-1}A)_{i,(j_1,j_2)}$. We first look at the case where $i=\top$. If $e_i$ denotes the appropriate $i$-unit vector, we have
\begin{eqnarray*}
    H_{\top, (j_1,j_2)}&:=&(-B^{-1}A)_{\top,(j_1,j_2)}=-e_{\top}^TB^{-1}Ae_{(j_1,j_2)}=c_{(j_1,j_2)}-c_B^TB^{-1}Ae_{(j_1,j_2)}\\
    &=&H_{1,(j_1,j_2)}
\end{eqnarray*}
using that $c_B=e_{\top}$ and $c_{(j_1,j_2)}=0$. The value of latter term we already know from (\ref{eq:reducedcost2}):
\[
    H_{\top, (j_1,j_2)}=H_{1,(j_1,j_2)}=t^{-w(j_1,j_2)}\walks^{\sigma}_{j_2\top}(t) -\walks^{\sigma}_{j_1\top}(t)
\]
and this is exactly what we needed to prove. Now assume $i\neq \top$. Then, using item 1 from \cref{thm:main} we get
\begin{eqnarray*}
    H_{i, (j_1,j_2)} &=& \sum_{v\in \vm\cup \{\top\}} -B^{-1}_{iv}A_{v(j_1,j_2)}\\
    &=& -B^{-1}_{i\top}A_{\top(j_1,j_2)}-B^{-1}_{ij_1}A_{j_1(j_1,j_2)}-\sum_{v\in \vm\backslash \{j_1\}} B^{-1}_{iv}A_{v(j_1,j_2)} \\
    &=& 0 \cdot A_{\top(j_1,j_2)} - \walks_{j_1i}^{\sigma}(t)\cdot (1-t^{-w(e)}\minwalks_{j_2,j_1}(t))\\
    & & - \sum_{v\in \vm\backslash \{j_1\}} \walks_{v,i}^{\sigma}(t)\cdot - t^{-w(e)}\minwalks_{j_2,v}(t)\\
    &=& t^{-w(e)}\left(\sum_{v\in \vm} \walks^{\sigma}_{v,i}(t)\cdot\minwalks_{j_2,v}(t)\right) -\walks^{\sigma}_{j_1,i}(t) \\
    &\stackrel{*}{=}& t^{-w(e)}\walks^{\sigma}_{j_2,i}(t)-\walks^{\sigma}_{j_1,i}(t)
\end{eqnarray*}
Here the equation marked with a star follows analogous to equation (\ref{eq:rewritewalks}): since $i\in\vm$, every walk from $j_2$ to $i$ contains some node of $\vm$. If $v$ denotes the first node from $\vm$ in a walk, then every $j_2,i$-walk in $G_{\sigma}$ can be split in a unique manner into a walk from $j_2$ to $v$ that uses only edges of $\emn$, and a walk in $G_{\sigma}$ from $v$ to $i$.  
\end{proof}

\cref{thm:interpretation} gives us some interpretations of tableau values and variable values. In particular, if we consider the exponents of leading terms of the series, then they can be expressed in terms of shortest paths in the LSP. 
However, given some BFS, the value of $x_B$ still depends on the choice of $b$, for which there could be multiple logical options. The simplest choice would be $b_i=1$ for all $i$, giving us $x_{(j_1,j_2)}=\sum_{v\in \vm}\walks^{\sigma}_{v,j_1}(t)$ for all $(j_1,j_2)\in \sigma$, and $z=\sum_{v\in \vm\cup \{\top\}}\walks^{\sigma}_{v,\top}(t)$. Another logical choice could be $b_i = \sum_{v'\in V} \minwalks_{v',i}(t)$, which yields
\begin{eqnarray*}
    x_{(j_1,j_2)}&=&\sum_{v\in \vm}\walks^{\sigma}_{v,j_1}(t)\sum_{v'\in V}\minwalks_{v',v}(t)\\
    &=&\sum_{v'\in V}\sum_{v\in \vm}\walks^{\sigma}_{v,j_1}(t)\minwalks_{v',v}(t)=\sum_{v'\in V}\walks^{\sigma}_{v',j_1}(t)
\end{eqnarray*}
and likewise $z=\sum_{v'\in V}\walks^{\sigma}_{v',\top}(t)$.

\subsection{Handling degeneracy}
In \cref{thm:main} we assume that the LSP is nondegenerate. In general, LSPs do not fulfill this property, even if they are constructed from a (sink) parity game or MPG. However, there are simple modifications possible to prevent degeneracy:
\begin{itemize}
    \item For a (sink) parity game, perform the following well-known preprocessing to give each vertex a unique priority: as long as there are two identical priorities $p$, one may simply fix one of these two, and then increase all other priorities in the game that are $\geq p$ by $2$. Doing this repeatedly gives an equivalent (sink) parity game.\footnote{In the resulting LSP, there could have been two strategies with equal valuation, but ties get broken with this preprocessing.} As a result, when transforming the (sink) parity game into an LSP, there will be no cycles of weight $0$ (except at $\top$). Moreover, for any two paths $P_1,P_2$ with the same start and end in the (sink) parity game, they can only have equal weight in the resulting LSP if they contained the same priorities, and thus the  exact same nodes. Since we do not allow for copied edges in our graphs, $P_1,P_2$ both must contain at least two edges, and that implies that the two paths must share an internal vertex. So we satisfy both conditions of nondegeneracy for LSPs.
    \item For a mean payoff game or LSP, multiply all weights by $2^{|V|}$. Label the vertices $v_1,v_2,\ldots, v_{|V|}$, and add a weight of $2^{i-1}$ to the outgoing edges of node $v_i$ for all $i$, except for $\top$. This does not change the sign of any cycle, except that cycles that had weight 0 now have a relatively small positive weight. For the second condition of nondegeneracy: any two distinct paths have different weights, unless the paths contain the exact same nodes. Again, then the paths must share an internal vertex, so we conclude that this modification is enough to guarantee a nondegenerate (induced) LSP. A similar technique is found in \cite{allamigeon_combinatorial_2014}.
\end{itemize}

In the proof of \cref{thm:main}, we only used the fact that the LSP is nondegenerate in the proof of item 3. Without this assumption, it could be that two BFSs have different objective values in the LP, while their related strategies have the same valuation. This can be because the weight of the shortest paths from a node to the sink in $G_{\sigma}$ can be the same, while the second/third/... shortest paths are different; this has no effect on strategy improvement, but does impact the LP.

One could think of the simplex algorithm in that case as emulating a variant of strategy improvement, where one considers a move improving if it lexicographically improves some sorted list of shortest paths to the sink. For example, a move that keeps the shortest path the same but improves the second shortest path would be improving. 

In theory, the steps the simplex algorithm takes in that case could still be computed, by expressing $\minwalks$ and $\walks$ as fractions of finite polynomials in $t$. That is because both are solutions to systems of linear equations in terms of finite polynomials in $t$. It is unclear, however, if the computations can be done efficiently: for example, even if $G_{\sigma}$ is acyclic, it could still contain exponentially many paths, thus $\walks^{\sigma}$ could contain exponentially many terms.

\section{Strategies and lopsided sets}\label{sec:lopsided}
We now use the established connection between strategy improvement and the simplex algorithm. In particular, we have a connection between strategies in games on graphs and vertices of polyhedra. In this section we show some properties of the set of admissible or winning strategies in these games. Our goal is to show they have a property called lopsidedness. First of all, it turns out that the polyhedron of the LP has a nice structure as follows.

\begin{definition}\label{def:hypercubelike}
    Consider $d\in \mathbb{N}$ and a polyhedron $P$, that is either empty, or pointed, simple and $d$-dimensional. We call such a polyhedron \emph{hypercubelike} if it is empty, or if its facets can be partitioned into $d$ sets $S_1,S_2,\ldots,S_d$, such that:
\begin{itemize}
    \item $|S_i|\in\{1,2\}$ for every $i$
    \item If $|S_i|=2$, then the two facets of $S_i$ have an empty intersection.
\end{itemize}
\end{definition}
Intuitively, a hypercubelike polyhedron looks like some part of a hypercube: there are $d$ directions, where in each direction there are (at most) two `opposite' facets. One may also view hypercubelike polyhedra as being prismatoids in multiple directions simultaneously, but ignoring parallelism. Or one could consider them 2-level polytopes, ignoring parallelism and boundedness. We analyze the structure of vertices of hypercubelike polyhedra. To do so, we consider them as realizations of certain sets.
\begin{definition}\label{def:b}
    Given $d\in \mathbb{N}$ and a hypercubelike polyhedron $P$, let $S_1,S_2,\ldots,S_d$ be a partition of its facets as in \cref{def:hypercubelike} (if it is nonempty). For every $i$ with $|S_i|=2$, we say $S_i=\{S_i^1,S_i^2\}$. Let $V(P)$ be the set of vertices of $P$, and let $b:V(P)\to \{-1,1\}^d$ be a function such that:
    \begin{itemize}
        \item We have $b(v)_i=b(v')_i$ for all $v,v'\in V(P)$ and all $i$ with $|S_i|=1$.
        \item We have $b(v)=-1$ if $v\in S_i^1$, and $b(v)=1$ if $v\in S_i^2$, for all $i$ with $|S_i|=2$.
    \end{itemize}
     Then $b(P)$ is defined by $\SetOf{b(v)}{v\text{ vertex of }P}$, where $b(P)=\emptyset$ if $P$ is empty.
\end{definition}
Since hypercubelike polyhedra are pointed and simple by definition (unless empty), the inclusion structure of their faces (face lattice) is completely determined by the set $b(P)$ \cite[Thm. 3.4]{joswig2001vertex}. Note that there can be many possible functions $b$ for one hypercubelike polyhedron $P$, depending on the chosen partition and labeling of the facets. 

We call a subset $C$ of $\{-1,1\}^d$ \emph{hypercube-realizable} if there exists a hypercubelike polyhedron $P$ and a function $b$ as in \cref{def:b} such that $C=b(P)$. Hypercube-realizability turns out to be a property that relates to many types of strategy sets. First, we need to assign subsets of $\{-1,1\}^d$ to games to be able to compare the two.

\begin{definition}\label[definition]{def:graphgame}
    Suppose we are given an LSP $\mathcal{G}$ where every Maximizer node has at most two outgoing edges. Let the nodes of $\vm$ be $v_1,v_2,\ldots,v_d$, and let the outgoing edge(s) of $v_i$ be labeled $e_i^1$ (and $e_i^2$). Given any Maximizer strategy, let $b(\sigma)\in \{-1,1\}^d$ be such that $b(\sigma)_i=-1$ if $e_i^1\in\sigma$ and $b(\sigma)_i=1$ otherwise. Let $b(\mathcal{G})=\SetOf{b(\sigma)}{\sigma \text{ admissible Maximizer strategy}}$.
    We define $b(\mathcal{G})$ similarly for other games on graphs, where every Player 0/Maximizer node has at most two choices:
    \begin{itemize}
        \item If $\mathcal{G}$ is a sink parity game or weak unichain MDP, \[b(\mathcal{G})=\SetOf{b(\sigma)}{\sigma \text{ admissible Player 0 strategy}}\]
        \item If $\mathcal{G}$ is a parity game, \[b(\mathcal{G})=\SetOf{b(\sigma)}{\sigma \text{ winning Player 0 strategy}}\]
        \item If $\mathcal{G}$ is an MPG, \[b(\mathcal{G})=\SetOf{b(\sigma)}{\sigma \text{ Maximizer strategy with average}\geq 0\text{ for every start }v_0}\] 
    \end{itemize}
\end{definition}

Here the assumption that Maximizer nodes have at most two outgoing edges is not very restrictive: if a game does not satisfy this condition, we can easily modify it into an equivalent game that does satisfy it, see the proof of \cref{cor:lawrencelopsided} and \cite[Fig. 4.14]{Friedmann2011ExponentialPrograms} for more detail. For MPGs one could also consider the sets of strategies which give valuation $\geq x$ for some fixed $x\in \mathbb{R}$. This is equivalent to \cref{def:graphgame}, because we can transform the game by subtracting $x$ from the reward of each edge/action, which reduces each valuation by $x$.

\begin{lemma}\label[lemma]{lem:gamesHR}
    Let $\mathcal{G}$ be a game on a graph out of the five defined in \cref{def:graphgame}. Then $b(\mathcal{G})$ is hypercube-realizable. 
\end{lemma}
\begin{proof}
    For LSP, this is a direct consequence of statement 2 from \cref{thm:main}. This then in turn proves the analogous statement for (sink) parity games and mean payoff games, because of their reductions to LSP.\footnote{One should be careful to not introduce new Maximizer edges in such a reduction. For example, the reduction from MPG to LSP can be done as follows: first add intermediate Minimizer nodes in all Maximizer-controlled positive-weight cycles, and then add an edge from every Minimizer node to the added sink node.}
    
    Analogously, for every weak unichain MDP, there is also an LP formulation, such that there is a one-to-one correspondence between admissible strategies and the vertices of the polyhedron of the LP: more details can be found in \cite{Friedmann2011ExponentialPrograms}.
\end{proof}

Before we continue to prove the main result of this section, we define the concept of multi-connectivity, which will help us with the proof. It is defined recursively as follows.
\begin{definition}\label[definition]{def:multiconnected}
        Consider a subset $C$ of $\{-1,1\}^d$, and consider $C$ to be a subset of vertices of the $d$-dimensional hypercube $[-1,1]^d$. Then we call $C$ \emph{multi-connected} if and only if the following two statements hold:
        \begin{enumerate}
            \item $C$ is empty, or there exists a $g\in C$ such that $C\backslash \{g\}$ is multi-connected
            \item For every face $F$ of the hypercube, all vertices of $C\cap F$ are connected by edges of $F$. Defined more precisely: for each pair of vertices $v,v'$ of $C\cap F$, there exists a $v,v'$-path in the vertex-edge graph of $F$, where each internal vertex of the path is also in $C\cap F$.
        \end{enumerate}
    \end{definition}
We prove a connection between hypercube-realizable and multi-connected sets. 
\begin{lemma}\label[lemma]{lem:HR-MC}
    Every hypercube-realizable set is multi-connected.
\end{lemma}
\begin{proof}
    Given some $d$, we consider $C\subseteq \{-1,1\}^d$ that is hypercube-realizable. We prove that the two properties of multi-connectivity hold for $C$, and do this by induction.

    As an induction basis, the empty set is both hypercube-realizable (realized by the empty polyhedron) and multi-connected (by definition). Likewise, if $|C|=1$, it is hypercube-realizable by a cone, and one can quickly verify that the properties of multi-connectivity hold.
    
    Next, suppose we have proven the two statements for all hypercube-realizable $|C|\leq m$ for some $m\geq 1$, and consider some hypercube-realizable $C\subseteq \{-1,1\}^d$ with $|C|=m+1$.

    Let $P$ be a hypercubelike polyhedron with $C=b(P)$. Let $v$ be any vertex of $P$, and let $c\in\mathbb{R}^d$ be such that $v$ is the unique minimizer of $c^Tx$ over $P$. Assume $c$ is chosen such that the value of $c^Tx$ is distinct in every vertex of $P$. Let $v'$ be the vertex with the largest value of $c^Tx$ out of all vertices of $P$. Note: $P$ may be unbounded, so $v'$ is not necessarily a maximizer of $c^Tx$. Because $m+1\geq 2$, we have $v'\neq v$. We can apply rotation, translation and scaling to $P$ such that $c=e_1$ (being the first unit vector in $\mathbb{R}^d$), and additionally $c^Tv'=0$ and $c^Tv=-1$. This means $-1<c^Tx<0$ for all other vertices of $P$. Let $Q=\{x\in P: c^Tx\leq 0\}$, this is a polytope whose vertices consist of the vertices of $P$, and possible some new vertices on the hyperplane $x_1=0$.

    Now we embed $Q$ into the hyperplane $x_{d+1}=1$ in $\mathbb{R}^{d+1}$, simply by setting the $d+1$-th coordinate to $1$ for each point (see \cref{fig:projectiontikz}). Let $\mathcal{C}$ be the cone in $\mathbb{R}^{d+1}$ given by~$\mathcal{C}=\SetOf{\alpha x}{x\in Q, \alpha\in \mathbb{R}_{\geq 0}}$. Let $Q'$ be the intersection of $\mathcal{C}$ with the hyperplane $x_1=-2$ in $\mathbb{R}^{d+1}$ (also in \cref{fig:projectiontikz}). 
    
    Because of how we defined it, $Q'$ can be interpreted as the result of a projective transform on $Q$ , which sends the hyperplane $c^Tx=0$ to `infinity.' Projective transforms preserve affine dependence, and because $Q$ lies on one side of the hyperplane $c^Tx=0$, convexity is preserved as well. It follows that $Q$ is a polyhedron, with the same face lattice as $Q$, except that the faces that were subsets of the hyperplane $c^Tx=0$ are not present anymore. As a result, the face lattice of $Q'$ is the same as that of $P$, except that all faces in which $c^Tx\geq 0$ for all $x$ are removed (these are precisely the faces whose \emph{only} vertex was $v'$).
    It follows that $Q'$ is hypercubelike, and we can pick $b(Q')$ such that $b(Q')=b(P)\backslash \{b(v')\}$. This shows that indeed there is a $g\in C$ that can be removed to obtain a multi-connected set, which is the first statement of the induction step.
\begin{figure}[htbp]
    \centering
    \resizebox{0.6\linewidth}{!}{%
    \begin{tikzpicture}
\begin{axis}
[grid=both,ymin=-0.5,ymax=4,xmax=1.5,xmin=-4,
               minor tick num=1,axis lines = middle,xlabel=$x_1$,ylabel=$x_2$,label style =
               {at={(ticklabel cs:1.1)}}]
    \addplot [gray!50,dashed,name path=A,domain=-4:0] {-x};
    \path [name path=B]
        (\pgfkeysvalueof{/pgfplots/xmin},4) --
        (\pgfkeysvalueof{/pgfplots/xmax},4);
     \addplot [gray!10] fill between [
        of=A and B,soft clip={domain=-4:-0.03},
    ];
      \addplot+[red,line width=2] coordinates {
        (-1,1) (0,1)
    };
    \addplot[black, dashed] coordinates {
        (-2,-1) (-2,5)
    };
    \addplot+[blue,line width=2] coordinates {
        (-2,2) (-2,5)
    };
    \node[red] at (axis cs:-1.15,0.8) {$v$};
    \node[red] at (axis cs:0.15,0.85) {$v'$};
    \node[red] at (-0.5,1.2) {$Q$};
    \node[blue] at (axis cs:-1.8,2.5) {$Q'$};
  \end{axis}
\end{tikzpicture}
    }
    \caption{The induction step illustrated for the case $d=1$ and $m=1$. In this case, $P$ is a line segment, and we have $P=Q$. The cone $\mathcal{C}$ is shaded in grey. After the `projection', we get the half-line $Q'$, which has one vertex less than $Q$.}
    \label{fig:projectiontikz}
\end{figure}
    
    For the second statement, we need to show that for every face $F$ of the $d$-dimensional hypercube, the vertices of $C\cap F$ are connected by edges. However, if $P$ again is a hypercubelike polyhedron with $b(P)=C$, then every face $F$ of the hypercube corresponds to a (possibly empty) face $F'$ of the polyhedron $P$. It is well-known that the edge-vertex graph of any face of a polyhedron is connected. This follows for example from the correctness of the simplex algorithm performed on that face: every vertex of $F'$ has a path in $F'$ towards the optimal vertex. Thus, every two vertices of $F$ are connected via (bounded) edges, and every bounded edge in $F'$ corresponds to an edge in $F$ whose endpoints are in $C$. It then follows that the vertices of $C\cap F$ must all be connected as well. This completes the second statement of the induction step.
\end{proof}

Now we are ready for the main goal of this section, to show that the sets of strategies of all the considered games have a property called lopsidedness. This property was introduced to describe properties of which orthants a convex set can occupy. The following definition (also called total asymmetry) is due to \cite[Thms. 3,4]{lawrence_lopsided_1983}

\begin{definition}\label[definition]{def:totalasymmetry}
    Consider a subset $C$ of $\{-1,1\}^d$, and consider $C$ to be a subset of vertices of the $d$-dimensional hypercube $[-1,1]^d$. For every face $F$ of the hypercube, let $f_F:\{-1,1\}^d\to \{-1,1\}^d$ be the function (involution) that maps each vertex in $F$ to its opposite vertex in $F$, and leaves all other vertices the same. 

    We call $C$ \emph{lopsided} if for every $F$ for which $f_{F}(C)=C$, we have that either no vertices of $F$ are in $C$, or all vertices of $F$ are in $C$.
\end{definition}

It turns out that lopsided sets have a very rich combinatorial structure. For example, \cite{bandelt_combinatorics_2006} describes 30 equivalent definitions of lopsidedness.

\begin{lemma}\label[lemma]{lem:MC-LS}
    Every multi-connected set is lopsided.
\end{lemma}
\begin{proof}
    Consider $C$ again as the set of vertices of the hypercube $[-1,1]^d$. A set $C\subseteq \{-1,1\}^d$ is called isometric if for every pair of elements $c,c'\in C$, there is a $c,c'$-path in the vertex-edge graph of the hypercube, consisting of only vertices of $C$ and of $d_{H}(c,c')$ edges, where $d_H$ is the Hamming distance.

    Isometricity follows from the second condition of multi-connectivity: suppose for contradiction that the second condition holds for $C$, and there is a pair of vertices $c,c'\in C$ where the shortest path between them has more than $d_H(c,c')$ edges. Assume that $c,c'$ are chosen such that $d_{H}(c,c')$ is minimal. If $F$ is the smallest face of the hypercube $[-1,1]^d$ that contains both $c$ and $c'$, then we must have $\dim(F)\geq 2$ (otherwise $c,c'$ would be neighbors). By the second condition of \cref{def:multiconnected} there is a path in $F$ (with only vertices of $F\cap C$) connecting $c$ and $c'$, and in particular there is a shortest such path. Let $c''\in F\cap C\backslash\{c,c'\}$ be in a shortest path. This path consists of a shortest $c,c''$-path, say $n_1$ edges long, and of a shortest $c'',c'$-path of $n_2$ edges, where $n_1+n_2>d_H(c,c')$. But since $d_H(c,c'')+d_H(c'',c')=d_H(c,c')$ by definition of $F$, we have $n_1>d_H(c,c'')$ or $n_2>d_H(c'',c')$, both of which contradict the assumption that $c,c'$ had the shortest possible Hamming distance. This completes the proof of isometricity.

    Then by \cref{def:multiconnected}, for any multi-connected set $C$, there exists a chain $C_0=\emptyset\subset C_1\subset C_2\subset\ldots\subset C_{|C|}=C$, such that the size of each set is one more than the previous set, and each set is isometric. It then follows from \cite[Prop. 4.2]{CHALOPIN20221} that $C$ is lopsided.
\end{proof}
The converse of \cref{lem:MC-LS} is not true. It follows from \cite[Prop. 4.2]{CHALOPIN20221} that every multi-connected set admits a so-called corner peeling, while not every lopsided set allows for a corner peeling (\cite[Thm. 4.5]{CHALOPIN20221} features a lopsided set without corner peeling with $d=12$).
   
Now, if we combine \cref{lem:gamesHR,lem:HR-MC,lem:MC-LS}, this gives us the main theorem of this section.
\begin{theorem}\label{thm:gameslopsided}
    Let $\mathcal{G}$ be any LSP, MPG, (sink) parity game or weak unichain MDP where every Maximizer/Player 0 node has at most two outgoing edges. Then the set $b(\mathcal{G})$ is lopsided. 
\end{theorem}

From the many equivalent definitions of lopsidedness, we now get a number of statements about strategy sets for these games. For example, we gain the following results.

\begin{corollary}\label[corollary]{cor:lawrencelopsided}
    Let $\mathcal{G}$ be an LSP, and let $S\subseteq \vm$. Then one of the following is true:
    \begin{itemize}
        \item There is a partial positional strategy $\sigma_S$ on $S$ (i.e. Maximizer only specifies his choices on $S$) such that every strategy extending $\sigma_S$ is admissible.
        \item There is a partial strategy $\tilde{\sigma}_{S}$ on $\vm\backslash S$ such that no strategy extending it is admissible.
    \end{itemize}
\end{corollary}
\begin{proof}
 We use \cref{thm:gameslopsided}, but first we need to make sure that every Maximizer node has outdegree at most two. To do so, for every Maximizer node $v$ with outdegree $k>2$ and successors $v_1',v_2',\ldots,v_k'$, we replace $v$ by the gadget in \cref{fig:outdegtwogadget}. Every Maximizer strategy in the new game corresponds to a Maximizer strategy in the original game: the token is sent from $v_1$ to some $v_i'$ in the new game, and we equivalently pick edge $(v,v_i')$ in the original game. The strategy in the new game is admissible if and only if the corresponding strategy in the original game is admissible. Thus, we only need to show the statement for the new game. The statement then follows immediately from \cref{thm:gameslopsided} and from the original definition of lopsidedness \cite{lawrence_lopsided_1983}.
\begin{figure}[htbp]
    \centering
    \resizebox{\linewidth}{!}{%
    \begin{tikzpicture}[node distance = 3cm,
    p0/.style={circle, draw=p0Blue!60, fill=p0Blue!40, thick, minimum size=15mm,inner sep=0pt, outer sep=0pt}, 
    every path/.style={draw=p0Blue,line width=1.3pt,{>=Latex}}, 
      p1/.style={rectangle, draw=p1Red!60, fill=p1Red!40, thick, minimum size=11mm}, 
      undef/.style={diamond, draw=black, fill=black!20, thick, minimum size=11mm},
      phan/.style={draw=none, fill=none, inner sep=0pt, outer sep=0pt}] 
    {\large
    \node[p0] (a1) at (0,0) {$v_1$};
    \node[p0] (a2) at (4,0) {$v_2$};
    \node[phan] (phan1) at (8,0) {}; 
    \node[phan] (phan2) at (12,0) {}; 
    \node[p0] (akm1) at (16,0) {$v_{k-1}$};
    \node[undef] (b1) at (0,4) {$v_1'$};
    \node[undef] (b2) at (4,4) {$v_2'$};
    \node[undef] (bkm2) at (12,4) {$v_{k-2}'$}; 
    \node[undef] (bkm1) at (16,4) {$v_{k-1'}$}; 
    \node[undef] (bk) at (20,0) {$v_k'$};

    \draw (a1) edge[->] (a2);
    \draw (a2) edge[->] (phan1);
    \draw ($(phan1)!0.3!(phan2)$) edge[-,black,dotted] ($(phan1)!0.7!(phan2)$);
    \draw ($(phan2)!0.2!(akm1)$)  edge[->] (akm1);
    \draw (akm1) edge[->] node [below] {$w(v,v_k')$} (bk);

    \draw (a1) edge[->] node[left] {$w(v,v_1')$} (b1);
    \draw (a2) edge[->] node[right] {$w(v,v_2')$} (b2);
    \draw ($(phan2)!0.2!(bkm2)$)  edge[->] node[right] {$w(v,v_{k-2}')$} (bkm2);
    \draw (akm1) edge[->] node[right]{$w(v,v_{k-1}')$} (bkm1);
    }
\end{tikzpicture}}
    \caption{The gadget replacing the node $v$. Circles represent Maximizer nodes, diamonds are nodes from the original game. Edge weights are shown, where no label means weight $0$. All edges previously going to $v$ now go to $v_1$ instead, and $v_1,\ldots,v_{k-1}$ are in $S$ if and only if $v$ was in $S$}
    \label{fig:outdegtwogadget}
\end{figure}
\end{proof}

\begin{corollary}\label[corollary]{cor:cornercube}
    Let $\mathcal{G}$ be a parity game where Player 0 can win from every node, and where every Player 0 node has outdegree at most $2$. For any strategy $\sigma$ and any set of Player 0 edges $S$ with $S\cap \sigma=\emptyset$, let $\sigma[S]$ be defined by
    \[
    \sigma[S](v)=\begin{cases} v' & (v,v')\in S\\
        \sigma(v) & \text{otherwise}
    \end{cases}
    \]
    Then, there exists a Player 0 strategy $\sigma$ with the following properties:
    \begin{itemize}
        \item $\sigma$ wins for Player 0 from every starting node.
        \item If $S=\SetOf{e\in V_{0}\backslash\sigma}{\sigma[\{e\}] \text{ is winning from every node}}$, and $S'\subseteq S$, then $\sigma[S']$ is winning from every node.
    \end{itemize}
\end{corollary}
\begin{proof}
    It follows from the proof of \cref{lem:MC-LS} and \cite[Prop 4.2]{CHALOPIN20221} that multi-connected sets allow for a corner peeling. In particular, every multi-connected set has a corner. Then by \cref{lem:gamesHR,lem:HR-MC} the set $b(\mathcal{G})$ also has a corner, and we let $\sigma$ be a strategy such that $b(\sigma)$ is a corner of $b(\mathcal{G})$. By \cite[Lemma 4.1(i,iii)]{CHALOPIN20221}, the smallest cube $F[t]$ containing a corner $t$ and its `neighbors' in a lopsided set $C$ is contained in the set $C$, which is equivalent to the statement of the corollary.
\end{proof}

\section{Discussion and future work}\label{sec:discussion}
We showed that for each nondegenerate LSP, there exists a linear program such that each run of the simplex algorithm corresponds to a run of strategy improvement, and vice versa. This creates a powerful new method to perform worst-case analysis for the simplex algorithm. There is one point of caution there: in general, the size of the coefficients of the LP will be doubly exponential due to the large coefficient~$t$. This means that this method cannot be naively used to show long runtime in terms of input size and (Turing) polynomial running time, but only in terms of number of input variables, and in terms of determining if a strongly polynomial algorithm exists.

This also implies that the LP described cannot be implemented naively, as the coefficients are too large. While for the nondegenerate case one can, of course, simulate a run of the simplex algorithm with an LSP, it is unclear if there is any way to efficiently compute improving moves in the degenerate case. 

Because the LP proposed in this paper is very closely related to Puiseux series, it is also conceivable that there is a variant of the tropical simplex method for which a similar tropical LP can be formulated: where the tropical simplex method would be equivalent to strategy improvement.

To show lopsidedness of the structure of strategy sets of games, we used hypercubelike polyhedra to show multi-connectivity. This highlights the polyhedral structure of the LPs considered. In some cases it is also possible to prove multi-connectivity directly for some games: one could show the first property of \cref{def:multiconnected} by slowly decreasing edge weights (or other variables, whenever the outcome continuously depends on the inputs). The second property can be shown using correctness of strategy improvement (currently it follows from correctness of the simplex method). With this technique, one can likely also show lopsidedness of strategy sets for other games where strategy improvement works.

There is some conceptual similarity between the hypercube-realizable sets defined in this paper and `realizable' sets as defined in \cite{lawrence_lopsided_1983}. That is since they both relate to convex shapes and are lopsided. However, it is shown in \cite{maat_discrete_2026} that the two notions are not the same.

We assumed a bound on the outdegree of Maximizer vertices in \cref{sec:lopsided}. As mentioned, this assumption is not restricting. However, if one does not make this assumption, the polyhedra considered will not resemble hypercubes, but instead resemble products of simplices. In this case, one can still show that the vertices of the polyhedra and the strategy sets correspond to ample sets, if one considers a more general notion of ampleness, which is introduced in \cite{chepoi2026ample}.


\section*{Acknowledgements}
Thanks to Georg Loho for guidance and countless helpful suggestions, and thanks to Nils Mosis for the discussions that inspired some of the results in this paper. Thanks also to Michael Joswig and Victor Chepoi for helpful comments, the latter in particular for showing how to simplify the proof of \cref{lem:MC-LS}.

%

\bibliography{references}   

%
%

\end{document}